\documentclass[a4paper]{article}

\usepackage{bm}
\usepackage{amsthm,amsmath,amssymb}
\usepackage{cite}
\usepackage[sans]{dsfont}
\usepackage[mathscr]{euscript}

\newcommand{\rmd}{\mathrm{d}}
\newcommand{\rme}{\mathrm{e}}
\newcommand{\rmi}{\mathrm{i}}
\newcommand{\openone}{\mathds{1}}
\newcommand{\Tr}{\operatorname{Tr}}

\newcommand{\RE}{\operatorname{Re}}

\newcommand{\Ebb}{\mathbb{E}}
\newcommand{\Rbb}{\mathbb{R}}
\newcommand{\Cbb}{\mathbb{C}}
\newcommand{\Qbb}{\mathbb{Q}}
\newcommand{\Pbb}{\mathbb{P}}

\newcommand{\norm}[1]{\left\Vert#1\right\Vert}
\newcommand{\abs}[1]{\left\vert#1\right\vert}

\newcommand{\Ical}{\mathcal{I}}

\newcommand{\Kcal}{\mathcal{K}}
\newcommand{\Lcal}{\mathcal{L}}

\newcommand{\Pcal}{\mathcal{P}}

\newcommand{\Tcal}{\mathcal{T}}

 \newcommand{\Cscr}{\mathscr{C}}

\newcommand{\Fscr}{\mathscr{F}}
\newcommand{\Gscr}{\mathscr{G}}
\newcommand{\Hscr}{\mathscr{H}}

\newcommand{\Lscr}{\mathscr{L}}

\newcommand{\Sscr}{\mathscr{S}}
\newcommand{\Tscr}{\mathscr{T}}

\newcommand{\Xscr}{\mathscr{X}}

\theoremstyle{plain}
\newtheorem{theorem}{Theorem}
\newtheorem{proposition}[theorem]{Proposition}
\theoremstyle{definition}
\newtheorem{assumption}{Assumption}
\theoremstyle{remark}
\newtheorem{remark}{Remark}

\begin{document}

\title{Jump-diffusion unravelling of a non Markovian generalized Lindblad master
equation}

\author{A. Barchielli \\ Politecnico di Milano, Dipartimento di Matematica, \\ Piazza Leonardo da
Vinci 32, I-20133 Milano, Italy \\ and Istituto Nazionale di Fisica Nucleare (INFN),
Sezione di Milano
\\
\\
C. Pellegrini \\
Laboratoire de Statistique et Probabilit\'es, Universit\'e Paul Sabatier, \\ 118,
Route de Narbonne, 31062 Toulouse Cedex 4, France}

\maketitle

\begin{abstract}
The ``correlated-projection technique'' has been successfully applied to derive a
large class of highly non Markovian dynamics, the so called non Markovian
generalized Lindblad type equations or Lindblad rate equations. In this article,
general unravellings are presented for these equations, described in terms of
jump-diffusion stochastic differential equations for wave functions. We show also
that the proposed unravelling can be interpreted in terms of measurements continuous
in time, but with some conceptual restrictions. The main point in the measurement
interpretation is that the structure itself of the underlying mathematical theory
poses restrictions on what can be considered as observable and what is not; such
restrictions can be seen as the effect of some kind of superselection rule. Finally,
we develop a concrete example and we discuss possible effects on the heterodyne
spectrum of a two-level system due to a structured thermal-like bath with memory.
\end{abstract}


\\ Schr\"odinger equation; heterodyne detection; direct detection; quantum
\\ trajectories \& memory; time-continuous measurements.

\section{Introduction}

Open quantum system theory concentrates on the study of the time evolution a quantum
system in contact with an environment; in particular, this theory aims to describe
phenomena such as decoherence, relaxation, emission of light, evolution of
entanglement \cite{francesco2,carm1,carm2,gardiner,W3}. Starting from the
Hamiltonian approach describing the coupled evolution of the quantum system and the
environment, the reduced evolution of the quantum system is obtained by tracing out
the degrees of freedom of the environment. This allows to describe the time
evolution of the open system in terms of its density matrix $\rho_S(t)$ with the
help of a quantum master equation. Invoking standard assumptions as weak coupling
limit and Born-Markov approximation, one can derive the Markovian quantum master
equation \cite{gardiner,francesco2,W3}, with infinitesimal generator in the Lindblad
form \cite{Gor,Lind}. This approach, called the \textit{Markovian approach}, is
physically based upon the absence of memory effects in the action of the
environment. This is a good and useful assumption in several physical examples,
namely in quantum optics \cite{francesco2,carm1,carm2,gardiner,W3}.

However, such assumptions are not valid in general and in many physically important
cases the description of a reduced quantum evolution requires a non-Markovian
approach involving strong and long memory effects. For example, situations with
strong coupled systems, entanglement and correlation in the initial state, finite
reservoirs\ldots need to be described by non-Markovian dynamics. Different
techniques, such as the Nakajima-Zwanzig projection technique, the
time-convolutionless operator technique, random Lindblad operator, random functional
equations have been developed to derive non-Markovian quantum master equations
\cite{BPP,gardiner,francesco2,Man,Br2,Diosi4,Diosi2,Vac,Lidar}. Recently, the
concept of correlated projection technique has been used in order to describe a
non-Markovian generalization of Lindblad type master equations (or Lindblad rate
equations) \cite{bud,Br,Br1}. This approach has been successfully applied to
describe non-Markovian models: structured reservoirs, two-state systems coupled with
energy bands \cite{Esp1,Esp2,bud2,BrGe,bud,Br,Bud08,Bud09,BudG09,Br1,Gem}\ldots

An active line of research concentrates on the study of the behaviour of the
solutions of these equations (thermalization, return to equilibrium,
decoherence,...). But even in the Markovian case, the quantum master equations
remain often of a formal interest. In particular, most of the equations cannot be
solved analytically and involve a large number of parameters which prevent numerical
simulations. Concerning the numerical aspect, a powerful approach is the theory of
``stochastic wave function unravelling''. This consists in constructing a stochastic
differential equation for a wave function $\psi(t)$ such that
$\mathbb{E}[\vert\psi(t)\rangle\langle\psi(t)\vert]=\rho_S(t)$. Then, by taking the
average of a large number of realizations of $\psi(t)$ one reproduces the solution
of the master equation. This has been applied in many Markovian situations
\cite{MolCD,francesco2,carm2,BarGreg09}. Concerning the non-Markovian framework,
different extensions of this approach have been developed \cite{PILO,Br3,BAR3}
(there is no general and common approach).

In the Markovian case the stochastic unravelling of the master equation has not only
a technical usefulness, but it can be also interpreted in terms of measurements in
continuous time; often the name of \emph{quantum trajectory theory} is used
\cite{BarGreg09,carm2}. In particular, for quantum optical systems the stochastic
formulation is used to describe direct, heterodyne and homodyne detection. However,
in the non-Markovian setup the notion of quantum trajectories as well as the
measurement interpretation are still highly debated
\cite{PILO,Diosi1,Diosi2,Diosi4,GW5,GW6,BAR3}.

For the non Markovian generalization of Lindblad type master equations
\cite{bud,Br,Br1}, only particular unravellings have been presented \cite{MoPe,XHX}.
In this article, we aim to present a general approach to obtain unravellings for
this type of equations and to show that in this case an interpretation of the
unravelling in terms of measurements in continuous time is possible. Our approach is
based upon the general technique used to unravel Markovian Lindblad equations. In
particular, our results include and generalize the previous results \cite{MoPe}.
However, we have an important conceptual difference from the Markovian case. We are
assuming that the structure of the bath responsible of the non Markovian behaviour
is not observable and this makes unobservable some of the components of the noises
introduced in the unravelling.

The article is structured as follows. In Section \ref{nMLME}, we describe the
Lindblad rate equations. In Section \ref{unrME}, we present the jump-diffusion
unravellings of these equations. In particular, we derive non Markovian
generalizations of stochastic Schr\"odinger equations.  The stochastic master
equations and the measurement interpretation are given in Section \ref{NMGSME}. In
Section \ref{appl}, we construct a concrete non Markovian model (a two level system
in contact with a structured environment), we present a possible unravelling, and we
show possible effects of the non Markov dynamics on the heterodyne spectrum.

\section{Non-Markovian generalized Lindblad-type \\ master
equations}\label{nMLME}

In this section we introduce the non Markovian Lindblad-type master equation which
we are interested in. These equations can be obtained by the application of the
\emph{correlated projection technique} and are sometimes called \emph{Lindblad rate
equations} \cite{Esp1,Esp2,Br2,bud2,BrGe,bud,Br,Bud08,Bud09,BudG09,Br1}. For any
separable complex Hilbert space $\Hscr$ we denote by $\Lscr(\Hscr)$ the space of the
linear bounded operators on $\Hscr$, by $\Tscr(\Hscr)$ the space of trace class
operators and by $\Sscr(\Hscr)$ the set of statistical operators (a statistical
operator is a trace class, positive operator with trace $1$).

Let $\Hscr_S$ denote the Hilbert space representing the open system. The generalized
master equation we consider is the evolution equation
\begin{equation}\label{eqbreuer1}
\frac{\rmd \ }{\rmd t}\,\eta_{i}(t)= -\rmi [H^i,\eta_{i}(t)] +\sum_{\alpha\in
A}\sum_{k=1}^n\left( R^{ik}_\alpha \eta_{k}(t) {R^{ik}_\alpha}^*- \frac 1 2
\left\{{R^{ki}_\alpha}^* R^{ki}_\alpha, \eta_{i}(t)\right\}\right)
\end{equation}
for the vector $\big(\eta_{1}(t),\ldots,\eta_{n}(t)\big)$ with components in
$\Tscr(\Hscr_S)$. The quantities $H^i,\, R^{ki}_\alpha$ are system operators which
we take to be bounded for mathematical simplicity and $A$ is a finite set of
indices.

\begin{assumption}\label{ass:1} $\Hscr_S$ is a complex separable Hilbert space,
$H^i = {H^i}^*\in \Lscr(\Hscr_S)$, $R^{ki}_\alpha\in \Lscr(\Hscr_S)$,
$k,i=1,\ldots,n$, $\alpha\in A$. The initial condition of Eq.\ \eqref{eqbreuer1} has
the properties
\begin{equation}\label{cqstate}
\eta_{i}(0)\in\Tscr(\Hscr_S), \quad \eta_{i}(0)\geq 0, \quad
\sum_{i=1}^n\Tr_{\Hscr_S} \left\{\eta_{i}(0)\right\}=1.
\end{equation}
\end{assumption}

\begin{remark} Equation \eqref{eqbreuer1} preserves the properties \eqref{cqstate} at all times
\cite{Br}; then, we interprete as system state the statistical operator
\begin{equation}\label{sysstat}
\eta_S(t)=\sum_{i=1}^n\eta_{i}(t).
\end{equation}
\end{remark}

The proof of the positivity preservation property of Eq.\ \eqref{eqbreuer1} is very
instructive and goes through the embedding of the dynamics $\{\eta_i(0)\} \mapsto
\{\eta_i(t)\}$ into an usual Lindblad dynamics in an extended state space \cite{Br}.
Let us consider the enlarged space $\Hscr=\Hscr_S\otimes\mathbb{C}^n$. Let
$\{e_i,\;i=1,\ldots,n\}$ be a reference orthonormal basis of $\mathbb{C}^n$. Let us
introduce the block diagonal operator $\tilde\eta(t)$ on $\Hscr$ by
\[
\tilde\eta(t)=\sum_{i=1}^n\eta_{i}(t)\otimes |e_i\rangle \langle e_i|
\]
and set
\begin{equation}\label{HS}
H=\sum_{i=1}^n H^i\otimes
|e_i\rangle
\langle e_i|, \qquad S^{ij}_\alpha = R^{ij}_\alpha \otimes |e_i\rangle \langle e_j|.
\end{equation}
Then, we get immediately from \eqref{eqbreuer1} the evolution equation for
$\tilde\eta(t)$:
\begin{equation}\label{Lindbext}
\frac{\rmd \ }{\rmd t}\,\tilde\eta(t)=\tilde\Lcal[\tilde\eta(t)] \equiv -\rmi
[H,\tilde\eta(t)] +\sum_{\alpha\in A}\sum_{i,j=1}^n \left(S^{ij}_\alpha
\tilde\eta(t) {S^{ij}_\alpha}^*- \frac 1 2
\left\{{S^{ij}_\alpha}^*{S^{ij}_\alpha},\tilde\eta(t)\right\}\right).
\end{equation}
For block diagonal initial conditions the two equations \eqref{eqbreuer1} and
\eqref{Lindbext} are completely equivalent. Let us note that the linear map
$\tilde\Lcal$ is explicitly in the Lindblad form, so that the maps
$\tilde\eta(0)\mapsto \tilde\eta(t)$ and $\{\eta_i(0)\}_{i=1}^{\;n} \mapsto
\{\eta_i(t)\}_{i=1}^{\;n}$ are \emph{completely positive} (CP).

In spite of the construction above, the index $i$ is not interpreted as a quantum
degree of freedom, but as the value of a classical observable. In typical
applications the index $i$ labels the energy bands of a structured environment
\cite{bud,MoPe,Gem,BrGe,Br}. A vector of operators with the properties
\eqref{cqstate} can be seen as a classical/quantum state. If we set
$p_i(t)=\Tr_{\Hscr_S}\{\eta_i(t)\}$ and $\hat \eta_i(t)=\eta_i(t)/p_i(t)$, we have
$\hat \eta_i(t)\in \Sscr(\Hscr_S)$, $p_i(t)\geq 0$, $\sum_{i=1}^np_i(t)=1$. In
quantum information the set of probabilities and statistical operators
$\left\{p_i(t), \hat\eta_i(t);\,i=1,\ldots , n\right\}$ is called an \emph{ensemble}
and it is completely equivalent to the vector
$\big(\eta_{1}(t),\ldots,\eta_{n}(t)\big)$ \cite{BarL05,HolS05}. In this setup the
system state \eqref{sysstat} is known as average state and it does not contain the
information on the classical label $i$. Equation \eqref{eqbreuer1} gives a
memoryless evolution for the ensemble $\left\{p_i(t), \hat\eta_i(t);\,i=1,\ldots ,
n\right\}$; it is the evolution of the system state $\eta_S(t)$ which is non
Markovian.

In the Markov case it is well known how to construct general unravellings of a
master equation and how to give a measurement interpretation to them. So, to have an
usual master equation in Lindblad form extending our non Markovian dynamics
\eqref{eqbreuer1} is a good starting point for the whole construction. However, Eq.\
\eqref{Lindbext} is not the unique extension of \eqref{eqbreuer1} and here we give
another extension which is in some sense more convenient as starting point. The
possible extensions depend on having or not the condition $R^{ij}_\alpha\propto
\delta_{ij}$; so, we put in evidence some diagonal terms.

\begin{assumption}\label{ass:3}
Let us take $A=\{-m_1,\ldots,-1,1,\ldots,m_2\}$,  and $
R^{ij}_{-\alpha}=\delta_{ij}L^i_\alpha$, for $ \alpha=1,\ldots,m_1$.
\end{assumption}

With this assumption Eq.\ \eqref{eqbreuer1} becomes
\begin{subequations}\label{eqbreuer2}
\begin{equation}
\frac{\rmd \ }{\rmd t}\,\eta_{i}(t)=\mathcal{K}_i\big(\eta_1(t),\ldots,\eta_n(t)\big),
\end{equation}
\begin{multline}\label{K_i}
\mathcal{K}_i\big(\tau_1,\ldots,\tau_n\big) := -\rmi [H^i,\tau_{i}]
+\sum_{\alpha=1}^{m_1}\left( L^{i}_\alpha \tau_{i} {L^{i}_\alpha}^*- \frac 1 2
\left\{{L^{i}_\alpha}^* L^{i}_\alpha, \tau_{i}\right\}\right)
\\ {}+\sum_{\alpha=1}^{m_2}\sum_{k=1}^n\left( R^{ik}_\alpha \tau_{k}
{R^{ik}_\alpha}^*- \frac 1 2 \left\{{R^{ki}_\alpha}^* R^{ki}_\alpha,
\tau_{i}\right\}\right).
\end{multline}
\end{subequations}

By using the operators \eqref{HS}, we define the new operators
\begin{subequations}\label{SS}
\begin{equation}
V_\alpha=\sum_{i,j=1}^n S^{ij}_{-\alpha} \equiv \sum_{i=1}^n L^{i}_\alpha \otimes |e_i\rangle \langle e_i|,
\qquad \alpha=1,\ldots,m_1\,,
\end{equation}
\begin{equation}
S^j_\beta=\sum_{i=1}^n S^{ij}_\beta \equiv \sum_{i=1}^nR^{ij}_\beta \otimes |e_i\rangle \langle e_j|,
\qquad \beta=1,\ldots,m_2\,,
\end{equation}
\end{subequations}
and the Lindblad map $\Lcal$: $\forall\tau\in\Tscr(\Hscr)$,
\begin{equation}\label{L}\begin{split}
\Lcal[\tau]=-\rmi [H,\tau]&+\sum_{\alpha=1}^{m_1} \left(V_\alpha \tau {V_\alpha}^*-
\frac 1 2 \left\{{V_\alpha}^*{V_\alpha},\tau\right\}\right)
\\ {}&+\sum_{\alpha=1}^{m_2}\sum_{j=1}^n \left(S^{j}_\alpha \tau {S^{j}_\alpha}^*- \frac 1 2
\left\{{S^{j}_\alpha}^*{S^{j}_\alpha},\tau\right\}\right).\end{split}
\end{equation}
Then, we consider the Markovian quantum master equation
\begin{equation}\label{ME}
\frac{\rmd \ }{\rmd
t}\,\eta(t)=\Lcal[\eta(t)],
\end{equation}
with the initial condition
\begin{equation}\label{eq:incond}
\eta(0)\in\Sscr(\Hscr),  \quad
\Tr_{\Cbb^n}\left\{\eta(0)\left(\openone\otimes |e_i\rangle\langle
e_i|\right)\right\}=\eta_i(0).
\end{equation}

\begin{remark}
Let us use a subscript $i$ to denote the $i$-th block on the diagonal of any
trace-class operator, i.e. $\tau_i=\Tr_{\Cbb^n}\left\{\tau\left(\openone\otimes
|e_i\rangle\langle e_i|\right)\right\}$. It is easy to check that
\begin{equation}
\tilde\Lcal[\tau]_i=\Lcal[\tau]_i=\mathcal{K}_i(\tau_1,\ldots,\tau_n),
\end{equation}
and, so, both the master equations \eqref{Lindbext} and \eqref{ME} reduce to the
same Lindblad rate equation \eqref{eqbreuer2} for the blocks on the diagonal, while
they are different for the off-diagonal blocks. Being equal at time $t=0$ due to
\eqref{eq:incond}, we have that the blocks on the diagonal of $\eta(t)$ are exactly
the quantities $\eta_i(t)$ satisfying Eq.\ \eqref{eqbreuer2}.
\end{remark}

Another way to describe the situation is to say that there is a \emph{superselection
rule} and only block-diagonal observables are permitted. Then, statistical operators
with the same blocks on the diagonal are equivalent and represent the same physical
state. In this sense the two master equations \eqref{Lindbext} and \eqref{ME} are
physically equivalent.

It is worthwhile to note that the operator $\tilde\Lcal$ can always be written in
the form $ \Lcal$. It is enough to change the meaning of the subscript in the
operators $R^{ij}_\alpha$ or $L^i_\alpha$ in such a way that it includes also the
index $i$. Then, given two triples $(i,j,\alpha)$ and $(i',j',\alpha')$, we have
that $i\neq i' \Rightarrow \alpha\neq \alpha'$ (the same holds for two couples
$(i,\alpha)$ and $(i',\alpha')$). In this way, in the sums in Eqs.\ \eqref{SS} only
one term survives and $\tilde\Lcal=\Lcal$. So, there is no loss of generality in
considering only the master equation \eqref{ME}; the other case is always included,
eventually at the price of a renaming and reordering of the indices.

It is useful to formalize the framework we have presented in terms of normal states
on $W^*$-algebras and of CP dynamics.
\begin{remark}\label{cqalgebras}
Let $\Cscr\big(\Xscr;\Lscr(\Hscr_S)\big)$ be the $W^*$-algebra of the functions from
$\Xscr=\{1,2,\ldots,n\}$ into $\Lscr(\Hscr_S)$ \cite{BarL05,BarL06}. By natural
identifications we have $\Cscr(\Xscr;\Cbb)\simeq \Cbb^n$ and
$\Cscr\big(\Xscr;\Lscr(\Hscr_S)\big)\simeq \Lscr(\Hscr_S)\otimes \Cbb^n$, so that
$a\in \Cscr\big(\Xscr;\Lscr(\Hscr_S)\big)$ means $a=(a_1,\ldots,a_n)$, $a_j\in
\Lscr(\Hscr)$; then, $\norm{a}=\max_{j\in \Xscr}\norm{a_j}$. The predual space of
$\Cscr\big(\Xscr;\Lscr(\Hscr_S)\big)$ is $\Cscr\big(\Xscr;\Tscr(\Hscr_S)\big)\simeq
\Tscr(\Hscr_S)\otimes \Cbb^n$, so that $\tau\in \Cscr\big(\Xscr;\Tscr(\Hscr_S)\big)$
means $\tau=(\tau_1,\ldots,\tau_n)$, $\tau_j\in \Tscr(\Hscr)$; then,
$\norm{\tau}_1=\sum_{j\in \Xscr}\norm{\tau_j}_1 = \sum_{j=1}^n
\Tr_{\Hscr_S}\left\{\sqrt{\tau_j^{\,*}\tau_j}\right\}$. In a natural way $a$ and
$\tau$ can be considered as block-diagonal elements of $\Lscr(\Hscr)$ and
$\Tscr(\Hscr)$, respectively: $a\simeq \sum_{j=1}^n a_j\otimes |e_j\rangle \langle
e_j|$, $\tau\simeq \sum_{j=1}^n \tau_j\otimes |e_j\rangle \langle e_j|$.
\end{remark}
\begin{remark}\label{reduceddyn}
Equations \eqref{L} and \eqref{ME} define a CP quantum dynamical semigroup
$\Tcal(t)$ on $\Tscr(\Hscr)$. Then, we define the projection $\Pcal : \Tscr(\Hscr)
\to \Cscr\big(\Xscr;\Tscr(\Hscr_S)\big)\subset \Tscr(\Hscr)$ by $(\Pcal[\tau])_j=
\Tr_{\Cbb^n}\{\tau (\openone \otimes |e_j\rangle \langle e_j|\}$. The dynamics
associated to the Lindblad rate equation \eqref{eqbreuer2} turns out to be $\Pcal
\circ \Tcal(t)\big|_{\Cscr\big(\Xscr;\Tscr(\Hscr_S)\big)}$; it is CP and Markovian.
Finally, we define the projection $\Pcal_S : \Cscr\big(\Xscr;\Tscr(\Hscr_S)\big) \to
\Tscr(\Hscr)$ by $\Pcal_S[\tau]= \sum_j \tau_j$. The CP dynamics giving the system
state \eqref{sysstat} is $\Pcal_S\circ\Pcal \circ
\Tcal(t)\big|_{\Cscr\big(\Xscr;\Tscr(\Hscr_S)\big)}$ and it is this dynamics which
is non Markovian.
\end{remark}

\section{Unravelling of non Markovian Lindblad-type master equations}\label{unrME}

In this section, we derive a general form of jump-diffusion stochastic differential
equations (SDEs) for wave functions in the enlarged space
$\Hscr=\Hscr_S\otimes\mathbb{C}^n$ which provide unravellings of the Lindblad rate
equations \eqref{eqbreuer2}. Having at hand the usual Markovian master equation
\eqref{ME}, we adopt the usual approach \cite{BarGreg09,BarPZ98,BarP96} of
stochastic Schr\"odinger equations in the Markovian case. This method is based on
classical stochastic calculus (see for instance Refs.\ \cite{STOD,Met82} and
\cite[Appendix A]{BarGreg09}) and the notion of \emph{a posteriori} states
\cite{Bel88,Bel89,Bel89c,Bel93}.

The key point of the theory is the construction of a linear and a non-linear
\emph{stochastic Schr\"odinger equation} (SSE), connected by a normalization and a
Girsanov transformation, and, then, of the linear and non-linear \emph{stochastic
master equations}. The non-linear SSE is the key starting point for numerical
simulations of the solution of a master equation, while the possibility of passing
to linear equations is fundamental for the possibility of giving a measurement
interpretation to the whole construction without violating the rules of quantum
mechanics. Finally, the non-linear stochastic master equation gives the \emph{a
posteriori states}, the conditional state to be attributed at the system at time
$t$, knowing the results of the measurement up to time $t$.

\subsection{The linear stochastic Schr\"odinger
equation}\label{stononmark}

We consider a filtered probability space
$\big(\Omega,\Fscr,(\Fscr_t),\mathbb{Q}\big)$, satisfying the \emph{usual
hypotheses} \cite[Appendix A]{BarGreg09}. On this space, we consider $d_1+d_2\times
n$ independent standard Wiener processes $W_\alpha$, $W_\beta^j$ ($\alpha=1,\ldots ,
d_1\leq m_1$; $\beta=1,\ldots, d_2\leq m_2$; $j=1,\ldots,n$) and
$(m_1-d_1)+(m_2-d_2)\times n$ independent standard Poisson point processes
$N_\alpha$ of intensity $\lambda_\alpha>0$ and $N_\beta^j$ of intensity
$\lambda_\beta^j>0$ ($\alpha=d_1+1,\ldots, m_1$; $\beta=d_2+1,\ldots, m_2$;
$j=1,\ldots,n$), also independent of the Wiener processes. All these processes are
adapted and $W_\alpha^k(t)$, $ N_\alpha(t)-\lambda_\alpha t$ and $
N_\alpha^j(t)-\lambda_\alpha^jt$ are $(\Fscr_t)$-martingales, under the reference
probability $\mathbb{Q}$ \cite{STOD,Met82}. The trajectories of the Wiener processes
are taken to be continuous and the trajectories of the Poisson processes continuous
from the right. We set also \[
\lambda=\sum_{\alpha=d_1+1}^{m_1}\lambda_\alpha+\sum_{\alpha=d_2+1}^{m_2}\sum_{j=1}^n\lambda_\alpha^j\,.
\]

Now, on $\big(\Omega,\Fscr,(\Fscr_t),\mathbb{Q}\big)$, we consider the following SDE
for an $\Hscr$-valued process:
\begin{multline}\label{LSSE}
\rmd\zeta(t)=\left(K+\frac \lambda 2 \right)\zeta(t_{-})\rmd t
+\sum_{\alpha=1}^{d_1}V_\alpha\zeta(t_{-})\rmd W_\alpha(t)
\\ {}+\sum_{\alpha=1}^{d_2}\sum_{k=1}^nS^{k}_\alpha\zeta(t_{-})\rmd W_\alpha^k(t) +
\sum_{\alpha=d_1+1}^{m_1}\left(\frac 1{\sqrt{\lambda_\alpha}}\,V_\alpha-\openone
\right)\zeta(t_{-} )\rmd N_\alpha(t) \\ {}+
\sum_{\alpha=d_2+1}^{m_2}\sum_{k=1}^n\left(\frac
1{\sqrt{\lambda_\alpha^k}}\,S_\alpha^k-\openone \right)\zeta(t_{-} )\rmd
N_\alpha^k(t),
\end{multline}
where the operator in the drift part is given by
\begin{equation*}
K=-\rmi H-\frac{1}{2}\sum_{\alpha=1}^{m_1}{V_\alpha}^*
V_\alpha-\frac{1}{2}\sum_{\alpha=1}^{m_2}\sum_{k=1}^n{S_\alpha^{k}}^* S_\alpha^{k}
=\sum_{j=1}^n K^j\otimes |e_j\rangle \langle e_j|,
\end{equation*}
\[
K^j=-\rmi H^j-\frac{1}{2}\sum_{\alpha=1}^{m_1}{L_\alpha^{j}}^*
L_\alpha^{j}-\frac{1}{2}\sum_{\alpha=1}^{m_2}\sum_{k=1}^n{R_\alpha^{kj}}^*
R_\alpha^{kj}\,.
\]
By using the decomposition $\zeta(t)=\sum_{j=1}^n \zeta_j(t)\otimes e_j$, we get the
equivalent system of SDEs
\begin{multline}\label{LSSEi}
\rmd\zeta_j(t)=\left(K^j+\frac \lambda 2 \right)\zeta_j(t_{-})\rmd
t+\sum_{\alpha=1}^{d_1}L_\alpha^j\zeta_j(t_{-})\rmd W_\alpha(t) \\ {}+
\sum_{\alpha=d_1+1}^{m_1}\left(\frac 1{\sqrt{\lambda_\alpha}}\,L_\alpha^j-\openone
\right)\zeta_j(t_{-} )\rmd N_\alpha(t) + \sum_{\alpha=1}^{d_2}\sum_{k=1}^n
R^{jk}_\alpha\zeta_k(t_{-})\rmd W_\alpha^k(t)
\\
{}+\sum_{\alpha=d_2+1}^{m_2}\sum_{k=1}^n\left(\frac
1{\sqrt{\lambda_\alpha^k}}\,R_\alpha^{jk}\zeta_k(t_{-} )-\zeta_j(t_{-} ) \right)\rmd
N_\alpha^k(t).
\end{multline}
As usual the solutions of SDEs with jumps are taken to be continuous from the right
with left limits (c\`{a}dl\`{a}g processes); the notation $t_-$ means the left limit.

\begin{remark}\label{rem:trick}
If some of the operators $S$ in the jump part is zero, we eliminate its contribution
by taking the corresponding Poisson process with zero intensity, so that it is
almost surely 0 for all times. In other words, if we have $S^k_\alpha=0$ for some
$k$ and some $\alpha>d_2$, we take $\lambda^k_\alpha\downarrow 0$.
\end{remark}

\begin{assumption}\label{ass:2}
We take a random normalized initial condition: \ $\displaystyle
\zeta(0)=\zeta^0=\sum_{i=1}^n \zeta^0_i\otimes e_i$,\quad $\zeta^0$ is
$\Fscr_0$-measurable, \ $\displaystyle\Ebb_\Qbb\left[ \norm{\zeta^0}^2 \right]\equiv
\sum_{i=1}^n\Ebb_\Qbb\left[ \norm{\zeta^0_i}^2\right]=1$. To reproduce the initial
condition \eqref{cqstate} we ask also $\Ebb_\Qbb\left[|\zeta^0\rangle\langle
\zeta^0|\right]=\eta(0)$. Mean values of random operators are defined in weak sense.
\end{assumption}

Equation \eqref{LSSE} is a particular case of the equations studied in Refs.\
\cite{BAR3,BarPZ98}, so, we refer to those papers for the properties of its
solution, while all the results could be obtained by standard arguments in
stochastic calculus and the It\^o formula for continuous and jump processes
summarized by the \emph{It\^o table}
\begin{equation}\label{Ito}
\begin{split}
\rmd W_\alpha(t)\rmd W_\beta(t)=\delta_{\alpha \beta}\rmd t, \qquad
&\rmd W_\alpha^k(t)\rmd W_\beta^l(t)=\delta_{\alpha \beta}\delta_{kl}\rmd t,
\\
\rmd
N_\alpha(t)\rmd N_\beta(t)=\delta_{\alpha \beta}\rmd
N_\alpha(t),\qquad
&\rmd N_\alpha^i(t)\rmd N_\beta^j(t)=\delta_{\alpha \beta}\delta_{ij}\rmd
N_\alpha^i(t);
\end{split}
\end{equation}
all the other products are vanishing.

\begin{theorem}[\!{\cite[Prop.\ 2.1, Theor.\ 2.4, Prop.\ 3.2]{BAR3};
\cite[Theor.\ 1.1, Theor.\ 1.2]{BarPZ98}}] \label{MMart} Under Assumptions
\ref{ass:1} and \ref{ass:2}, the SDE \eqref{LSSE} admits a unique (up to
$\Qbb$-equivalence) solution $\zeta(t)$, $t\geq 0$. Moreover, the mean state
$\Ebb_\Qbb[|\zeta(t)\rangle \langle \zeta(t)|]$ satisfies the master equation
\eqref{ME}.

Finally, under the probability $\mathbb{Q}$, the process $\displaystyle p(t):=
\Vert\zeta(t)\Vert^2\equiv \sum_{i=1}^n \norm{\zeta_i(t)}^2$ is a non-negative
$(\mathscr{F}_t)$-martingale with $\Qbb$-mean $1$ and it satisfies the Dol\'eans SDE
\begin{multline}\label{dol}
\rmd p(t)= p(t_-)\biggl\{\sum_{\alpha=1}^{d_1} v_\alpha(t)\rmd W_\alpha(t)
+\sum_{\alpha=d_1+1}^{m_1}\left(\frac{ I_\alpha(t)}{
\lambda_\alpha}-1\right)\Big(\rmd N_\alpha(t)-\lambda_\alpha\rmd t\Big)
\\ {}+ \sum_{\alpha=1}^{d_2} \sum_{k=1}^n v_\alpha^k(t)\rmd W_\alpha^k(t)
+ \sum_{\alpha=d_2+1}^{m_2}\sum_{k=1}^n\left(\frac{ I_\alpha^k(t)}{
\lambda_\alpha^k}-1\right)\Big(\rmd N_\alpha^k(t)-\lambda_\alpha^k\rmd
t\Big)\biggr\} ,
\end{multline}
where
\begin{subequations}\label{quantity}
\begin{gather}\label{v_a}
v_\alpha(t)=2\RE\left\langle\psi(t_-)\Big|V _\alpha \psi(t_-)\right\rangle\equiv
2\sum_{j=1}^n\RE\left\langle\psi_j(t_-)\Big|L^{j} _\alpha \psi_j(t_-)\right\rangle,
\\ \label{vak}
v_\alpha^k(t)=2\RE\left\langle\psi(t_-)\Big|S^{k} _\alpha
\psi(t_-)\right\rangle\equiv 2\sum_{j=1}^n\RE\left\langle\psi_j(t_-)\Big|R^{jk}
_\alpha \psi_k(t_-)\right\rangle,
\\
I_\beta(t)=\norm{V_\beta\psi(t_-)}^2 \equiv
\sum_{j=1}^n\norm{L^{j}_\beta\psi_j(t_-)}^2,
\\
I_\beta^k(t)=\norm{S^{k}_\beta\psi(t_-)}^2 \equiv
\sum_{j=1}^n\norm{R^{jk}_\beta\psi_k(t_-)}^2.
\end{gather}
\end{subequations}
The process
\begin{subequations}\label{def:psi}
\begin{equation}
\psi(t) = \sum_{i=1}^n \psi_i(t)\otimes e_i
\end{equation}
is defined by
\begin{equation}\label{defpsik}
\begin{cases}
\displaystyle \psi_k(t)=\frac{\zeta_k(t)}{\norm{
\zeta(t)}}, &\textrm{if}\,\,\norm{\zeta(t)}\neq0,\\
\psi_k(t)=\psi, &\textrm{if}\,\,\norm{\zeta(t)}=0,
\end{cases}\end{equation}
\end{subequations}
where $\psi\in\Hscr_S$ is a fixed vector of norm $1/\sqrt{n}$.
\end{theorem}

\begin{remark}[A first unravelling]
By the theorem above, $\Ebb_\Qbb[|\zeta(t)\rangle \langle \zeta(t)|]$ satisfies the
master equation \eqref{ME} with initial condition $\eta(0)$ (Assumption
\ref{ass:2}). So, we have $\eta(t)=\Ebb_\Qbb[|\zeta(t)\rangle \langle \zeta(t)|]$,
$\forall t\geq 0$, and, by the discussion below Eq.\ \eqref{eq:incond}, we get
\begin{equation}\label{1unr}
\eta_i(t)=\Ebb_\Qbb[|\zeta_i(t)\rangle \langle \zeta_i(t)|], \qquad i=1,\ldots,n, \quad t\geq 0,
\end{equation}
which shows that $\zeta(t)$ is a pure-state unravelling of the solution of the
Lindblad rate equation \eqref{eqbreuer2}.
\end{remark}

\begin{remark}[\!{\cite[Theor.\ 1.2]{BarPZ98}; \cite[Theor.\ 29.2]{Met82}}]\label{solDol}
The solution of the Dol\'eans SDE \eqref{dol} is
\begin{multline*}
p(t)=\norm{\zeta^0}^2 \exp \bigg\{ \sum_{\alpha=1}^{d_1} \bigg( \int_0^t
v_\alpha(s)\rmd W_\alpha(s)- \frac 1 2 \int_0^t v_\alpha(s)^2 \rmd s \bigg)
\\ {}+ \sum_{\alpha=1}^{d_2} \sum_{k=1}^n\bigg( \int_0^t
v_\alpha^k(s)\rmd W_\alpha^k(s)- \frac 1 2 \int_0^t v_\alpha^k(s)^2 \rmd s
\bigg)\bigg\}
\\ {}\times \prod_{\beta=d_1+1}^{m_1}
\biggl\{\exp\bigg[\int_0^t\left(\lambda_\beta- I_\beta(s)\right)\rmd s \bigg]
 \prod_{r\in (0,t]}\biggl[1+ \biggl(\frac{I_\beta(r)}{\lambda_\beta} -1 \biggr)
\Delta N_\beta(r)\biggr]\biggr\}
\\ {}\times \prod_{\beta=m_1+1}^{m}\prod_{\ell=1}^n
\biggl\{\exp\bigg[\int_0^t\left(\lambda_\beta^\ell- I_\beta^\ell(s)\right)\rmd s
\bigg]
 \prod_{r\in (0,t]}\biggl[1+ \biggl(\frac{I_\beta^\ell(r)}{\lambda_\beta^\ell} -1 \biggr)
\Delta N^\ell_\beta(r)\biggr]\biggr\},
\end{multline*}
where $\Delta N_\beta(r,\omega)=N_\beta(r,\omega)- N_\beta(r_-,\omega)$, $\Delta
N_\beta^\ell(r,\omega)=N_\beta^\ell(r,\omega)- N_\beta^\ell(r_-,\omega)$. By the
fact that a Poisson process has only a finite number of jumps in a compact interval,
for every $\omega$ only a finite number of factors contributes to the product over
$r$ in the representation above.

Note that, if for some $t,\,\omega,\, \beta,\, \ell$ one has
$I_\beta^\ell(t,\omega)=0$ and $\Delta N^\ell_\beta(t,\omega)=1$, then
$p(T,\omega)=0$, $\forall T>t$. Similarly, $I_\beta(t,\omega)=0$ and $\Delta
N_\beta(t,\omega)=1$ imply $p(T,\omega)=0$, $\forall T>t$.
\end{remark}

\subsection{The generalized stochastic Schr\"odinger
equation} The final aim is to derive an equation for the normalized process
\eqref{def:psi}. This is based upon It\^o stochastic calculus again and a
Girsanov-type change of measure.

\begin{remark}[The change of probability measure]\label{physprob}
For for every $T>0$, we define the \emph{physical probability} $\Pbb^T$ over
$(\Omega, \Fscr_T)$ by
\begin{equation}\label{physprobPT}
\mathbb{P}^T(A)=\mathbb{E}_\mathbb{Q}\left[{1}
_Ap(T)\right]\equiv \int_A\Vert\zeta(T,\omega)\Vert^2 \Qbb(\rmd \omega), \qquad \forall A\in\mathscr{F}_T\,.
\end{equation}
Note that $\Pbb^T$ depends also on $\zeta^0$, which we assume to be normalized in
the sense of Assumption \ref{ass:2}. The martingale property given in Theorem
\ref{MMart} ensures that the family of probabilities $\{\Pbb^T,\,T>0\}$ is
consistent, that is
\begin{equation}\label{cons}
0<t<T, \quad A\in\mathscr{F}_t \qquad \Rightarrow \qquad
\mathbb{P}^T(A)=\mathbb{P}^t(A).
\end{equation}
\end{remark}

To obtain from \eqref{cons} the existence of a unique probability in the infinite
horizon limit $T\to +\infty$ is a delicate problem and can be guaranteed only with
respect to some sub-filtration composed by Borel standard $\sigma$-algebras
\cite[Section A.5.5]{BarGreg09}.

It is important to note that the denominator $\norm{\zeta(t)}$ in the definition of
the processes $\psi_k(t)$ could indeed vanish as stated in Remark \ref{solDol}. But,
by the construction in Remark \ref{physprob}, this happens with probability zero
with respect to the new probability $\mathbb{P}^T$, while this is not guaranteed
under the reference probability $\Qbb$.

The important consequences of this change of measure are the modification of the
characteristics of the driving processes $N_\alpha^k(t)$ and $W_\alpha^j(t)$ (due to
some extension of the Girsanov theorem to the diffusive/jump case \cite{STOD}) and
the fact that $\psi(t)$ satisfies a non linear SDE, the \emph{stochastic
Schr\"odinger equation} \cite{Bel88,Bel89,Bel89c,Bel93,BarGreg09}.

\begin{theorem}[\!{\cite[Prop.\ 2.5, Theor.\ 2.7]{BAR3};
\cite[Prop.\ 1.1, Theor.\ 1.3]{BarPZ98}}] Under the probability $\mathbb{P}^T$, the
processes $\hat{W}_\alpha$, $\hat{W}_\beta^k$, $t\in[0,T]$, $\alpha=1,\ldots, d_1$,
$\beta=1,\ldots, d_2$, $k=1,\ldots, n$, defined by
\begin{equation}\label{hatW}
\hat{W}_\alpha(t)=W_\alpha(t)-\int_0^t v_\alpha(s)\,\rmd s, \qquad
\hat{W}_\beta^k(t)=W_\beta^k(t)-\int_0^t v_\beta^k(s)\,\rmd s,
\end{equation}
are independent standard Wiener processes and the processes $N_\alpha(t)$,
$N_\beta^k(t)$, $t\in[0,T]$, $\alpha=d_1+1,\ldots, m_1$, $\beta=d_2+1,\ldots, m_2$,
$k=1,\ldots, n$,  are counting processes of stochastic intensities $I_\alpha(t)$ and
$I_\beta^k(t)$, respectively.

Again under the probability $\mathbb{P}^T$, the components of the process $\psi(t)$
satisfy in the time interval $[0,T]$ the SDE
\begin{subequations}\label{EDSNL}
\begin{multline}\label{dpsi}
\rmd \psi_j(t)=V_j\big(\psi_1(t_{-}),\ldots,\psi_n(t_{-})\big)\rmd t
+\sum_{\alpha=1}^{d_1}\Big(L^{j} _\alpha-\frac 1 2
\,v_\alpha(t)\Big)\psi_j(t_-)\rmd\hat{W} _\alpha(t)
\\
{}+\sum_{\alpha=1}^{d_2}\sum_{k=1}^n\Big(R^{jk} _\alpha\psi_k(t_-)-\frac 1 2
\,v_\alpha^k(t)\psi_j(t_-)\Big)\rmd\hat{W} _\alpha^k(t)
\\ {}+
\sum_{\alpha=d_1+1}^{m_1}\left(\frac{L^{j}_\alpha}{\sqrt{I_\alpha(t)}}-1
\right)\psi_j(t_{-})\,\rmd N_\alpha(t) \\ {}+ \sum_{\alpha=d_2+1}^{m_2}\sum_{k=1}
^n\left(\frac{R^{jk}_\alpha\psi_k(t_{-})}{\sqrt{I_\alpha^k(t)}}-\psi_j(t_{-}
)\right)\rmd N_\alpha^k(t),
\end{multline}
where
\begin{multline}\label{Vj}
V_j(\psi_1(t_-),\ldots,\psi_n(t_-))=K^j\psi_j(t_-)+\frac{1}{2}
\sum_{\alpha=d_1+1}^{m_1}I_\alpha(t)\psi_j(t_-)\\
{}+\frac{1}{2} \sum_{\alpha=d_2+1}^{m_2}\sum_{k=1}^nI^k_\alpha(t)\psi_j(t_-)+\frac 1
2 \sum_{\alpha=1}^{d_1}v_\alpha(t)\left(L^{j}_\alpha
-\frac{1}{4}\,v_\alpha(t)\right)\psi_j(t_-)
\\ {}
+\frac 1 2 \sum_{\alpha=1}^{d_2}\sum_{k=1}^nv_\alpha^k(t)\left(R^{jk
}_\alpha\psi_k(t_-) -\frac{1}{4}\,v_\alpha^k(t)\psi_j(t_-)\right).
\end{multline}
\end{subequations}
\end{theorem}

Note that the SDE \eqref{EDSNL} is non-linear in $\psi(t)$, because the quantities
$I_\alpha(t)$, $I^k_\alpha(t)$, $v_\alpha(t)$, $v_\alpha^k(t)$ are bilinear in
$\psi(t)$ itself. Moreover, to consider \eqref{EDSNL} as a closed equation for
$\psi(t)$ poses interesting mathematical problems on the definition of
\emph{solution} and on the meaning of uniqueness because the law of the driving
noises $N_\alpha^k$ depends on the solution $\psi(t)$ itself through the stochastic
intensities $I_\alpha$, $I_\alpha^k$ \cite{p3}.

\begin{proposition}[A normalized unravelling] The solution of the Lindblad rate equation
\eqref{eqbreuer2} can be expressed as the following mean with respect to the
physical probability
\begin{equation}\label{2unr}
\eta_i(t)=\Ebb_{\Pbb^T}[|\psi_i(t)\rangle \langle \psi_i(t)|], \qquad i=1,\ldots,n, \quad T\geq t\geq 0.
\end{equation}
\end{proposition}
\begin{proof}
Let us introduce the set $A_t=\{\omega\in \Omega: \norm{\zeta(t,\omega)}=0\}$. Then,
by the definitions of $p(t)$ and $\psi(t)$ given in Theorem \ref{MMart}, we have
\[
|\zeta_i(t)\rangle \langle \zeta_i(t)|={1}_{A_t^c}|\zeta_i(t)\rangle \langle \zeta_i(t)|
={1}_{A_t^c}p(t)|\psi_i(t)\rangle \langle \psi_i(t)|=p(t)|\psi_i(t)\rangle \langle \psi_i(t)|.
\]
By taking the $\Qbb$-expectation and by taking into account Eq.\ \eqref{1unr} and
the definition of the new probability, we get
\[
\eta_i(t)=\Ebb_\Qbb[|\zeta_i(t)\rangle \langle \zeta_i(t)|]=
\Ebb_\Qbb[p(t)|\psi_i(t)\rangle \langle \psi_i(t)|]=
\Ebb_{\Pbb^t}[|\psi_i(t)\rangle \langle \psi_i(t)|].
\]
Finally, by the consistency property \eqref{cons}, we get \eqref{2unr}.
\end{proof}

This proposition gives an unravelling of the Lindblad rate equation
\eqref{eqbreuer2} based on the components of the \emph{normalized} vector $\psi(t)$.
When $d_1=m_1=0$ and $d_2=0$, we recover the pure jump unravelling proposed in Ref.\
\cite{MoPe}. If the aim is only to simulate Eq.\ \eqref{eqbreuer2}, a normalized
pure state unravelling is much more efficient than a non-normalized one such as
\eqref{1unr} \cite{francesco2}. The simulation techniques based on \eqref{EDSNL}
with $d_1=d_2=0$ correspond to the Monte-Carlo wave function method started in Ref.\
\cite{MolCD}, while the case $d_1=m_1$ and $d_2=m_2$ gives rise to simulations of
diffusive type as in Refs.\ \cite{Gisin2,GisP92,GisKPTW93}. From the point of view
of simulations, the fact that the starting point was Eq.\ \eqref{ME}, and not Eq.\
\eqref{Lindbext}, has produced a more convenient unravelling with less noises (no
dependence on the label $j$).

\section{Measurements and stochastic master equations}\label{NMGSME}

In this section we face the problem of the measurement interpretation of the
unravelling we have constructed. We introduce the notions of \textit{instruments}
and \textit{a posteriori} states and we derive the non Markovian generalization of
the \emph{stochastic master equations}.

\subsection{Outputs and noises}
In the theory of measurements in continuous time
\cite{BarGreg09,BAR3,BarP96,BarPZ98} it is assumed that the output of the
measurement is given by some components of the driving noises appearing in the SDE
(Eq.\ \eqref{LSSE} or \eqref{LSSEi} in our case); the law of the output in $[0,T]$
is the physical probability \eqref{physprobPT}. Not all the components of $W$ and
$N$ have to contribute to the output. The role of some of the components of the
noises could be only to perform the unravelling of some dissipative term.

Let us examine first the components ${W}_\alpha^k(t)$, $\alpha= 1, \ldots, d_2$,
$k=1,\ldots,n$. For $t\in[0,T]$, under the physical probability $\Pbb^T$, from
\eqref{hatW} we get ${W}_\alpha^k(t)=\hat W_\alpha^k(t)+\int_0^t v_\alpha^k(s)\,\rmd
s$; but, as one sees from Eq.\ \eqref{vak}, $v_\alpha^k(s)$ mixes different
components of $\psi(t)$ and cannot be an observable, because it does not respect the
superselection rule. In particular the mean value of ${W}_\alpha^k(t)$ turns out to
be $\Ebb_{\Pbb^T}\left[{W}_\alpha^k(t)\right]=\int_0^t
\Ebb_{\Pbb^T}\left[v^k_\alpha(s)\right]\rmd s$ with
\begin{equation*}
\Ebb_{\Pbb^T}\left[v^k_\alpha(t)\right]=\Ebb_{\Pbb^t}\left[v^k_\alpha(t)\right]=
2\sum_{j=1}^n\RE \Tr_\Hscr\left\{\left(R^{jk}_\alpha\otimes
|e_j\rangle\langle e_k|\right)\eta(t)\right\}
\end{equation*}
and it involves the unphysical non-diagonal blocks
$\Tr_{\Cbb^n}\left\{\left(\openone\otimes |e_j\rangle\langle
e_k|\right)\eta(t)\right\}$. So, ${W}_\alpha^k(t)$ cannot contribute to the output.

No problem of this kind arises for the other processes, as one sees from Eqs.
\eqref{quantity}. The stochastic intensities $I_\alpha(t)$, $I^k_\alpha(t)$ and the
processes $v_\alpha(t)$ do not mix different components of $\psi(t)$. However, if
the counting process $N^k_\alpha$ is detected we gain information on the block
contributing to the emission (the block $k$), as one sees for instance from the mean
intensity
\begin{equation*}
\Ebb_{\Pbb^t}\left[I^k_\alpha(t)\right]=\sum_{j=1}^n
\Tr_{\Hscr_S}\left\{{R^{jk}_\alpha}^* R^{jk}_\alpha\eta_k(t)\right\}.
\end{equation*}
If we assume that the index $k$ is not physically observable coherently with the
fact that the system state is the sum \eqref{sysstat}, the process $N^k_\alpha$ is
not observable by itself. However, there is no obstruction in considering as
physically observable the counting process
\begin{equation}
M_\alpha(t):= \sum_{k=1}^n N^k_\alpha(t), \qquad \alpha=d_2+1,\ldots,m_2,
\end{equation}
whose stochastic intensity, under the physical probability, is $\sum_{k=1}^n
I^k_\alpha(t)$. No problem arises on the observability of the other counting
processes $N_\beta$ ($\beta=d_1+1,\ldots,m_1$), whose stochastic intensity under the
physical probability is $I_\beta(t)$.

Let us stress that, under the reference probability $\Qbb$, $M_\alpha$ is a Poisson
process of intensity
\begin{equation}\label{Lambda}
\Lambda_\alpha=\sum_{k=1}^n\lambda^k_\alpha\,, \qquad
\alpha=d_2+1,\ldots,m_2\,.
\end{equation}

Let us consider finally the processes $W_\alpha$. At least in quantum optical
systems, observations with a ``diffusive'' character come out from heterodyne or
homodyne detection and the involved operators must have an explicit time dependence
due to the presence of the \emph{local oscillator} \cite[Chapt.\ 7]{BarGreg09}. We
assume a very smooth time dependence which does not cause any essential change in
the previous results.

\begin{assumption}
For $\alpha=1,\ldots,d_1$, we assume the operators $L^{j}_\alpha$ to be time
dependent and given by
\[
L^{j}_\alpha(t)=\overline{ h^{j}_\alpha(t)}\, \hat L^{j}_\alpha\,, \qquad
\hat L^{j}_\alpha\in \Lscr(\Hscr_S),\quad \abs{h^{j}_\alpha(t)}=1;
\]
the complex functions $h^{j}_\alpha(t)$ are continuous from the left.
\end{assumption}

No time dependence is introduced into the master equations of Section \ref{nMLME}.
The explicit time dependence involves only the terms with $\rmd W_\alpha$ in Eqs.\
\eqref{LSSE}, \eqref{LSSEi}, \eqref{dpsi} and the third term in the right hand side
of \eqref{Vj}; moreover, from Eq.\ \eqref{v_a}, we get
\[
v_\alpha(t)=2\RE \sum_{j=1}^n\overline{h^j_\alpha(t)}\big\langle\psi_j(t_-)\big|\hat L^{j}
_\alpha \psi_j(t_-)\big\rangle.
\]

The key result of the previous discussion is that, due to the mathematical structure
and the meaning of the discrete label in the states, only the processes $W_\alpha$
($\alpha=1,\ldots, d_1$), $N_\beta$ ($\beta=d_1+1,\ldots,m_1$), $M_\gamma$
($\gamma=d_2+1,\ldots,m_2$) can be considered as possible components of the output.
However, some of the components could represent pure noises, not observed
quantities. So, we assume that only the first components are observed.

\begin{remark}\label{rem:obs_out}
Let us take $d_1'\leq d_1$, $m_1'\leq m_1$, $m_2'\leq m_2$. We assume that the
observed outputs are $W_\alpha$ with $1\leq \alpha \leq d_1'$, $N_\beta$ with
$d_1+1\leq \beta \leq m_1'$, $M_\gamma$ with $d_2+1\leq \gamma \leq m_2'$. If some
set of indices is empty, no component of the corresponding process is observed. The
law of the output is the physical probability \eqref{physprobPT}. Finally we denote
by $\{\Gscr_t,\,t\geq 0\}$ the augmented natural filtration generated by the set of
the observed processes.
\end{remark}

\subsection{The linear stochastic master equations and the instruments}
It is possible to have only the observed processes as driving noises in the
dynamical equations, but for this we need to work with density matrices and
trace-class operators. Let us introduce the positive trace-class operators
\begin{equation}\label{sigma}
\sigma(t):= \Ebb_\Qbb\big[|\zeta(t)\rangle\langle \zeta(t)|\big| \Gscr_t \big]
, \qquad \sigma_i(t)= \Ebb_\Qbb\big[|\zeta_i(t)\rangle\langle
\zeta_i(t)|\big| \Gscr_t \big].
\end{equation}
This means to take the mean on the non-observed components of the noises. Let us
recall that $\zeta(0)$ is connected to the initial condition $\eta(0)$ (given in
Eq.\ \eqref{eq:incond}) by Assumption \ref{ass:2}. By the fact that $\Gscr_0$ is
trivial we get
\begin{equation}\label{sigma0}
\sigma(0)=\eta(0)\in\Sscr(\Hscr), \qquad\sigma_i(0)=\eta_i(0).
\end{equation}

\begin{proposition}
In the stochastic basis $(\Omega,\mathscr{F},\mathscr{G}_t,\mathbb{Q})$, the
operator valued process $\sigma(t)$ satisfies the linear stochastic master equation
\begin{multline}\label{lSME}
\rmd \sigma(t)=\Lcal[\sigma(t_-)]\rmd t +
\sum_{\alpha=1}^{d_1'}\left(V_\alpha(t)\sigma(t_{-})+\sigma(t_-)V_\alpha(t)^{*}\right)
\rmd W_\alpha(t)
\\ {}
+\sum_{\alpha =d_1+1}^{m_1'}\left(\frac{V_\alpha \sigma(t_{-})
V_\alpha^*}{\lambda_\alpha} -\sigma(t_{-})\right)(\rmd N_\alpha(t)-\lambda_\alpha
\rmd t)
\\ {}
+\sum_{\alpha =d_2+1}^{m_2'}\left(\sum_{k=1}^n\frac{S^{k}_\alpha \sigma(t_{-})
{S^{k}_\alpha}^*}{\Lambda_\alpha} -\sigma(t_{-})\right)\left(\rmd
M_\alpha(t)-\Lambda_\alpha \rmd t\right),
\end{multline}
where $\mathcal{K}_i\big(\rho_1,\ldots,\rho_n\big)$ is defined by Eq.\ \eqref{K_i},
$\Lscr$ by \eqref{L}, $\Lambda_\alpha$ by \eqref{Lambda},
\begin{equation*}
V_\alpha(t)= \sum_{i=1}^n \overline{h^{i}_\alpha(t)}\, \hat L^{i}_\alpha \otimes |e_i\rangle \langle e_i|,
\qquad \alpha=1,\ldots,d_1.
\end{equation*}
Given the initial condition, Eq.\ \eqref{lSME} has pathwise unique solution. For the
components (the blocks on the diagonal) Eq.\ \eqref{lSME} reduces to
\begin{multline}\label{xx}
\rmd\sigma_j(t)=\mathcal{K}_j\big(\sigma_1(t_{-}),\ldots,\sigma_n(t_{-} )\big)\rmd t
\\ {}+\sum_{\alpha=1}^{d_1'}\left(\overline{h^j_\alpha(t)}\, \hat
L^{j}_\alpha\sigma_j(t_{-})+h^j_\alpha(t)\sigma_j(t_-){\hat L^{j*}_\alpha}\right)
\rmd W_\alpha(t)
\\ {}
+\sum_{\alpha =d_1+1}^{m_1'}\left(\frac{L^{j}_\alpha \sigma_j(t_{-})
{L^{j*}_\alpha}}{\lambda_\alpha} -\sigma_j(t_{-})\right)(\rmd
N_\alpha(t)-\lambda_\alpha \rmd t)
\\ {}
+\sum_{\alpha =d_2+1}^{m_2'}\left(\sum_{k=1}^n\frac{R^{jk}_\alpha \sigma_k(t_{-})
{R^{jk}_\alpha}^*}{\Lambda_\alpha} -\sigma_j(t_{-})\right)\left(\rmd
M_\alpha(t)-\Lambda_\alpha \rmd t\right).
\end{multline}
\end{proposition}
\begin{proof}
By applying the It\^o formula to $|\zeta(t)\rangle\langle \zeta(t)|$ and, then, by
taking the conditional expectation, we get the linear stochastic master equation for
$\sigma(t)$ as explained in \cite[Sect.\ 4.2]{BAR3}. Existence and uniqueness of the
solution of Eq.\ \eqref{lSME} is given in \cite[Prop. 3.4]{BAR3}. Equation
\eqref{xx} is obtained by direct computations.
\end{proof}

\begin{remark} Let us consider now the physical probability introduced in Remark
\ref{physprob}. The probability density of the restriction of $\Pbb^t$ to $\Gscr_t$
with respect to the reference measure $\Qbb$ on $\Gscr_t$ is
\begin{equation}
p_\Gscr(t)= \Ebb_\Qbb[p(t)|\Gscr_t]=\Tr_\Hscr\{\sigma(t)\}\equiv \sum_{j=1}^n
\Tr_{\Hscr_S}\{\sigma_j(t)\}.
\end{equation}
The density $p_\Gscr$ is a $\Gscr$-martingale under $\Qbb$ and the restrictions of
the physical probabilities are consistent.
\end{remark}

In the axiomatic formulation of a quantum theory, measurements are represented by
instruments, which give the probabilities and the states  after the measurement
(\emph{a posteriori} states). As in Refs.\ \cite{BarPZ98,BAR3}, we put
\begin{equation}
\mathcal{I}_t(F)[\eta(0)]=\mathbb{E}_\mathbb{Q}[{1}_F\sigma(t)], \qquad
F\in\Gscr_t\,, \quad \eta(0)\in \Sscr(\Hscr).
\end{equation}
By linearity we extend $\mathcal{I}_t(F)$ to the whole $\Tscr(\Hscr)$ and we get an
\emph{instrument} with value space $(\Omega,\Gscr_t)$, which means that $\Ical_t(F)$
is a CP map from $\Tscr(\Hscr)$ into itself for all $F\in \Gscr_t$, it is a strongly
$\sigma$-additive measure as a function of $F$ and $\Ical_t(\Omega)$ is
trace-preserving.

\begin{remark} Let us particularize the definition of instrument in the enlarged space to our case.
We define
\begin{equation}
\mathcal{I}_t^i(F)[\eta_1(0),\ldots,\eta_n(0)]=\mathbb{E}_\mathbb{Q}[{1}_F\sigma_i(t)], \qquad
F\in\Gscr_t,
\end{equation}
for all $\eta(0)$ satisfying the superselection rules. With the notations of Remarks
\ref{cqalgebras} and \ref{reduceddyn} we have $\big( \Ical_t^1(F),\ldots,
\Ical_t^n(F)\big)=\Pcal\circ \Ical_t(F)\big|_{\Cscr\big(\Xscr;\Tscr(\Hscr_S)\big)}$.
This is an instrument with the same value space as before, but made up of maps on
$\Cscr\big(\Xscr;\Tscr(\Hscr_S)\big)$. Finally, by defining
\begin{equation}
\Ical_t^S(F)=\sum_{j=1}^n \Ical^j_t(F),
\end{equation}
we get an instrument with value space $(\Omega,\Gscr_t)$ made up of CP maps from
$\Cscr\big(\Xscr;\Tscr(\Hscr_S)\big)$ into $\Tscr(\Hscr_S)$. Moreover, the
connection with the various dynamical maps introduced in Remark \ref{reduceddyn} is
given by $\Ical_t(\Omega)=\Tcal(t)$,
\[
\big(\Ical^1_t(\Omega),\ldots,\Ical^n_t(\Omega)\big)=
\Pcal\circ \Tcal(t)\big|_{\Cscr\big(\Xscr;\Tscr(\Hscr_S)\big)}, \quad \Ical^S_t(\Omega)
=\Pcal_S\circ\Pcal\circ \Tcal(t)\big|_{\Cscr\big(\Xscr;\Tscr(\Hscr_S)\big)}.
\]
\end{remark}

The instruments give the physical probabilities once one has the pre-mea\-surement
state. In our case we have, $\forall F\in \Gscr_t$,
\begin{equation}\begin{split}
\Pbb^t(F)&=\Tr_\Hscr\{\Ical_t(F)[\eta(0)]\}=
\sum_{i=1}^n\Tr_{\Hscr_S}\{\mathcal{I}_t^i(F)[\eta_1(0),\ldots]\}\\ {}&=
\Tr_{\Hscr_S}\{\mathcal{I}_t^S(F)[\eta_1(0),\ldots]\}.\end{split}
\end{equation}
This equation says that the probabilities $\Pbb^t$, introduced before by starting
from some stochastic differential equation, can be obtained also from instruments;
so, the axiomatic structure of a quantum theory is respected and the interpretation
as physical probabilities is justified.

\subsection{The a posteriori states and the stochastic master equation}
The instruments give also the \emph{a posteriori states}, the conditional states
after the measurement. Let us recall the definition in the case of $\Ical_t$; in the
other cases the definition is analogue. The \emph{a posteriori} state for the
instrument $\Ical_t$ and the pre-measurement state $\eta(0)$ is the
$\Sscr(\Hscr)$-valued random variable $\rho(t)$ such that
\[
\Ical_t(F)[\eta(0)]=\Ebb_{\Pbb^t}[{1}_F\rho(t)], \qquad \forall F\in \Gscr_t.
\]
By taking into account that the density of $\Pbb^t$ with respect to $\Qbb$ is the
trace of $\sigma(t)$ and how $\Ical_t(F)[\eta(0)]$ is defined in terms of
$\sigma(t)$, we get easily
\[
\rho(t)=
\frac {\sigma(t)}{\Tr_\Hscr\{\sigma(t)\}}=
\Ebb_{\Pbb^t}\big[\vert\psi(t)\rangle\langle\psi(t)\vert\big|\Gscr_t\big].
\]
The components of $\rho(t)$, which are
\[
\rho_{i}(t)=\Ebb_{\Pbb^t}\big[\vert\psi_i(t)\rangle\langle\psi_i(t)\vert\big|\Gscr_t\big]=
\frac {\sigma_i(t)}{\Tr_\Hscr\{\sigma(t)\}},
\quad i=1,\ldots,n,
\]
give the \emph{a posteriori} states for $\Ical^i_t$:
\begin{equation}\label{rho_i}
\mathcal{I}_t^i(F)[\eta_1(0),\ldots,\eta_n(0)]=\mathbb{E}_{\Pbb^t}[{1}_F\rho_i(t)].
\end{equation}
Note that we have also $\Ebb_{\Pbb^t}[\rho_i(t)]=\eta_i(t)$. Finally, by taking the
sum over $i$ in Eq. \eqref{rho_i}, we get the \emph{a posteriori} states
$\rho_S(t)=\sum_i \rho_i(t)$ for $\Ical^S_t$.

On the other side, the states $\eta(t)$, $\eta_i(t)$ are called the \emph{a priori
states}, due to the fact that these states are the averages of the \emph{a
posteriori} states and that they are interpreted as the states to be assigned to the
system at time $t$ when the result of the observation is not taken into account.

\begin{remark}[The stochastic master equation]
For $\alpha=1,\ldots,d_1'$, $\beta=d_1+1,\ldots,m_1'$, $\gamma=d_2+1,\ldots,m_2'$,
let us define
\begin{equation}
m_\alpha(t)=:\Ebb_{\Pbb^t}\left[v_\alpha(t)\big|\Gscr_t\right]=2\RE\sum_{j=1}^n \overline{h^j_\alpha(t)}
\Tr_{\Hscr_S}\left\{\hat L^{j}_\alpha\rho_j(t_-)\right\},
\end{equation}
\begin{equation}
J_\beta^1(t):=\Ebb_{\Pbb^t}\left[I_\beta(t)\big|\Gscr_t\right]=\sum_{j=1}^n
\Tr_{\Hscr_S}\left\{{L^{j}_\beta}^* L^{j}_\beta\rho_j(t_-)\right\},
\end{equation}
\begin{equation}
J_\gamma^2(t):= \sum_{k=1}^n \Ebb_{\Pbb^t}\big[I^k_\gamma(t)\big|\Gscr_t\big]=\sum_{j,k=1}^n
\Tr_{\Hscr_S}\left\{{R^{jk}_\gamma}^* R^{jk}_\gamma\rho_k(t_-)\right\}.
\end{equation}
Then, by stochastic calculus, under the new probability $\Pbb^T$ and for
$t\in[0,T]$, we get the equation for $\rho(t)$ \cite[Rem. 3.6]{BAR3} and, then, the
\emph{stochastic master equation} for the components
\begin{multline}\label{xxyy}
\rmd\rho_j(t)=\mathcal{K}_j\big(\rho_1(t_{-}),\ldots,\rho_n(t_{-} )\big)\rmd t
\\ {}
+\sum_{\alpha=1}^{d_1'}\left(\overline{h^j_\alpha(t)}\, \hat
L^{j}_\alpha\rho_j(t_{-})+h^j_\alpha(t)\rho_j(t_-){\hat L^{j*}_\alpha}-
m_\alpha(t)\rho_j(t_-)\right) \rmd \hat W_\alpha(t)
\\ {}
+\sum_{\beta =d_1+1}^{m_1'}\left(\frac{L^{j}_\beta \rho_j(t_{-})
{L^{j*}_\beta}}{J_\beta^1(t)} -\rho_j(t_{-})\right)(\rmd N_\beta(t)-J_\beta^1(t)
\rmd t)
\\ {}
+\sum_{\gamma =d_2+1}^{m_2'}\left(\sum_{k=1}^n\frac{R^{jk}_\gamma \rho_k(t_{-})
{R^{jk}_\gamma}^*}{J_\gamma^2(t)} -\rho_j(t_{-})\right)\left(\rmd
M_\gamma(t)-J_\gamma^2(t) \rmd t\right).
\end{multline}
The processes $\hat W_\alpha$ are independent standard Wiener processes,
$N_\beta(t)$ is a counting process of stochastic intensity $J_\beta^1(t)$ and
$M_\gamma(t)$ is a counting process of stochastic intensity $J_\gamma^2(t)$.
\end{remark}

\section{A two-level system in a structured bath}\label{appl}

To give a simple, but concrete example of the theory we have developed and to have a
first idea of the effects on physically measurable quantities, here we study a model
of a two level atom in contact with a non-trivial structured reservoir and we
compute the heterodyne spectrum of its emitted light. This is a modification of a
model \cite{bud,Br,MoPe} which could represent the dynamics of a single qubit in a
non Markovian environment or the dynamics of an optically active molecule, as the
fluorophore system, in a local nano-environment \cite{Bud09}.

We consider a two-level system in contact with a two-band reservoir; so,
$\Hscr_S=\Cbb^2$ and $n=2$. Let $\sigma_z$, $\sigma_\pm$ be the usual Pauli
matrices; then, $P_+=\sigma_+\sigma_-$ is the projection on the excited state
$\binom{1}{0}$ and $P_-=\sigma_-\sigma_+$ the projection on the ground state
$\binom{0}{1}$. Here we give the mathematical model, while the physical
interpretation is given when we write down the various dynamical equations. By using
the notations introduced in Assumption \ref{ass:3} and Section \ref{stononmark}, the
model we consider is defined by the following choices:
\begin{equation}\label{eq:model}\begin{split}
& d_1=m_1=2, \quad d_2=0, \quad m_2=2; \qquad H^i=\frac{\omega_i}2\, \sigma_z\,, \quad
\omega_i>0, \quad i=1,2;
\\
& R_1^{ii}=0,  \quad R^{21}_1=\sqrt{\gamma_1}\, \sigma_-\,, \quad
R^{12}_1=\sqrt{\gamma_2}\, \sigma_+\,, \qquad \gamma_i>0, \quad i=1,2,
\\
& R_2^{ii}=0,   \quad   R_2^{12}=0, \quad R^{21}_2=\sqrt{\gamma_0\varkappa}\,
\openone, \qquad \varkappa> 0, \quad \gamma_0>0; \qquad 0<\epsilon\leq 1,
\\
& L^{1}_1(t)=L^{2}_1(t)=\rme^{\rmi \nu t}\sqrt{\gamma_0\epsilon}\, \sigma_-\,, \quad
L^{1}_2(t)=L^{2}_2(t)=\rme^{\rmi \nu t}\sqrt{\gamma_0(1-\epsilon)}\,
\sigma_-\,,\quad \nu\in \Rbb.
\end{split}\end{equation}
The driving processes in the linear SDEs are the standard Wiener processes $W_1$,
$W_2$ and the Poisson processes $N^1_1$, $N^2_1$, $N^1_2$, with intensities
$\lambda^1_1$, $\lambda^2_1$, $\lambda^1_2$; all these processes are independent.
According to Remark \ref{rem:trick}, take $\lambda^2_2\downarrow 0$, so that $N^2_2$
is almost surely $0$ and we can set $\rmd N^2_2(t)=0$.

\subsection{The Lindblad rate equation and the equilibrium state}

First of all let us write down in the concrete case introduced above the Lindblad
rate equation \eqref{eqbreuer2}
\begin{subequations}\label{model}
\begin{multline}
\frac{\rmd\ }{\rmd t}\,\eta_1(t)=\Kcal_1\big(\eta_1(t),\eta_2(t)\big)
\equiv\gamma_0\left(\sigma_-\eta_1(t)\sigma_+ -\frac{1}{2}\{ P_+,\eta_1(t)\}\right)
\\ {}+ \gamma_2\sigma_+\eta_2(t)\sigma_--\frac{\gamma_1}{2}\{ P_+,\eta_1(t)\}-
\gamma_0\varkappa\eta_1(t)- \frac{\rmi \omega_1}2\left[\sigma_z, \eta_1(t) \right] ,
\end{multline}
\begin{multline}
\frac{\rmd\ }{\rmd t}\,\eta_2(t)=\Kcal_2\big(\eta_1(t),\eta_2(t)\big) \equiv
\gamma_0\left(\sigma_-\eta_2(t)\sigma_+ -\frac{1}{2}\{ P_+,\eta_2(t)\}\right) \\ {}+
\gamma_1\sigma_-\eta_1(t)\sigma_+-\frac{\gamma_2}{2}\{
P_-,\eta_2(t)\}+\gamma_0\varkappa\eta_1(t)- \frac{\rmi \omega_2}2\left[\sigma_z,
\eta_2(t) \right] .
\end{multline}
\end{subequations}

The model of Refs.\ \cite{Br,BrGe} corresponds to $\gamma_0=0$, $\varkappa=0$,
$\omega_1=\omega_2$; moreover, the rotating framework is used, so that the terms
with $\omega_i$ disappear. In Refs.\ \cite{bud,Bud09} the case
$\omega_1\neq\omega_2$ is allowed and it is explained by different energy shifts
induced by the two bands of the environment. So, we have a two level molecule with
two resonance frequencies due to the structured environment. The terms with
$\gamma_1$ and $\gamma_2$ represent the molecular transitions induced by the
environment and concomitant with transitions between the two bands of the nano
environment.

Reference \cite{Bud09} studies the stimulated fluorescence light under laser
excitation of the molecule; the treatment is based on the quantum regression
formula. Instead, our aim is to study the spontaneously emitted light and to this
end we have added the first term in both equations, the one with $\gamma_0$, which
is an explicit spontaneous emission term.

To have emission without stimulation by external light, we need some thermal-like
excitation. To get this effect we have added the terms with $\gamma_0\varkappa$.
This is the simplest modification giving rise to a non trivial equilibrium state.

If we write Eqs.\ \eqref{model} in terms of the matrix elements of $\eta_1$ and
$\eta_2$ we get two decoupled systems of equations: the system for the coherences
(the off diagonal terms) and the system for the populations (the diagonal terms).
Firstly, one checks easily that the coherences decay exponentially to zero. On the
other side, the system of equations for the diagonal terms turns out to be
equivalent to a 4-state, irreducible classical Markov chain. If we denote by $1^+$,
$1^-$, $2^+$, $2^-$ the four states, the transition rates different from zero are
$\gamma_0$ for the transition $1^+\to 1^-$, $\gamma_1$ for the transition $1^+\to
2^-$, $\gamma_0\varkappa$ for $1^+\to 2^+$, $\gamma_2$ for $2^-\to 1^+$, $\gamma_0$
for $2^+\to 2^-$, $\gamma_0\varkappa$ for $1^-\to 2^-$. From the graph of this
finite-state Markov chain we see that it is irreducible; then, there is a unique
equilibrium distribution, computed below, and it is a global attractor.

\subsubsection*{Equilibrium state.}
The Lindblad rate equations \eqref{model} admit a unique equilibrium state
$\eta_i(\infty)=\lim_{t\to +\infty} \eta_i(t)$, $i=1,2$, which can be easily
computed. It turns out to be given by
\[
\eta_i(\infty)=p_i \left(z_i^+P_++ z_i^-P_-\right),
\qquad z_i^-:=1-z_i^+\,, \qquad z_i^+:=\frac {\varkappa_i} {1+\varkappa_i}\,,
\]
\[
\varkappa_1:=\varkappa , \qquad
\varkappa_2:=\frac{\gamma_2\varkappa}{\gamma_1+\gamma_0(1+\varkappa)}\,, \qquad
p_1:=p, \qquad p_2:=1-p, \]
\[
p:=\frac{\gamma_2(1+\varkappa)}{\gamma_2+
\varkappa\left(\gamma_0+\gamma_2+\gamma_1\right) + \varkappa^2
\left(\gamma_0+\gamma_2\right)}\,.
\]
Let us note that we have $ (1-p)z_2^+=\varkappa p z_1^+$. By recalling that the
system state is the sum of the components \eqref{sysstat}, we get that the average
equilibrium state of the two-level system is
\[
\eta_S^{\mathrm{ eq}}=\eta_1(\infty)+\eta_2(\infty)=p\varkappa
P_++(1-p\varkappa)P_-\,.
\]

\subsection{The stochastic Schr\"odinger equations}

The lSSE \eqref{LSSEi} corresponding to the choices \eqref{eq:model} is
\begin{multline*}
\rmd\zeta_1(t)=\left(K^1+\frac \lambda 2 \right)\zeta_1(t_{-})\rmd t
-\zeta_1(t)\bigl(\rmd N_1^1(t)+\rmd N_2^1(t)\bigr)
\\ {}+\left(\sqrt{\frac
{\gamma_2}{\lambda_2}}\,\sigma_+\zeta_2(t_{-} )-\zeta_1(t_{-} ) \right)\rmd N^2_1(t)
\\ {}+ \rme^{\rmi \nu
t}\sqrt{\gamma_0}\,\sigma_-\zeta_1(t_{-})\left(\sqrt{\epsilon}\,\rmd W_1(t)+
\sqrt{1-\epsilon}\,\rmd W_2(t)\right),
\end{multline*}
\begin{multline*}
\rmd\zeta_2(t)=\left(K^2+\frac \lambda 2 \right)\zeta_2(t_{-})\rmd t
+\left(\sqrt{\frac {\gamma_0\varkappa}{\lambda_0}}\,\zeta_1(t_{-} )-\zeta_2(t_{-} )
\right)\rmd N_2^1(t)\\
{}-\zeta_2(t)\rmd N_1^2(t)
 +\left(\sqrt{\frac {\gamma_1}{\lambda_1}}\,\sigma_-\zeta_1(t_{-} )-\zeta_2(t_{-}
) \right)\rmd N_1^1(t) \\ {}+ \rme^{\rmi \nu
t}\sqrt{\gamma_0}\,\sigma_-\zeta_2(t_{-})\left(\sqrt{\epsilon}\,\rmd W_1(t)+
\sqrt{1-\epsilon}\,\rmd W_2(t)\right),
\end{multline*}
where $\lambda=\lambda^1_1+ \lambda^2_1+ \lambda^1_2$ and
\[
K^1=-\frac{\rmi\omega_1}2\, \sigma_z- \frac{\gamma_0+\gamma_1}2\,P_+ -
\frac{\gamma_0\varkappa}2\,\openone,
\qquad
K^2=-\frac{\rmi\omega_2}2\, \sigma_z- \frac{\gamma_0}2\,P_+ -
\frac{\gamma_2}2\,P_-.
\]
Note that the Wiener processes $W_1$ and $W_2$ appear always in the combination
$\sqrt{\epsilon}\, W_1(t)+ \sqrt{1-\epsilon}\, W_2(t)$, which is again a
one-dimensional standard Wiener process. The reason for the introduction of two
components is that the diffusive term represents the emitted light, which we have
divided in two channels: channel 1, represented by $W_1$, contains the light
reaching the heterodyne detector and channel 2, represented by $W_2$, contains the
lost light. The proportion of lost light is $1-\epsilon$.

Finally, by Eqs.\ \eqref{EDSNL}, the SSE for the normalized vectors is, under the
physical probability,
\begin{multline*}
\rmd\psi_1(t)=V_1(\psi_1(t_-),\psi_2(t_-))\rmd t  -\psi_1(t)\bigl(\rmd N_1^1(t)+\rmd
N_2^1(t)\bigr)\\ {}+\left(\frac {\sigma_+\psi_2(t_{-} )}{\norm{\sigma_+\psi_2(t_{-}
)}} -\psi_1(t_{-} ) \right)\rmd N_1^2(t)
\\
{}+ \sqrt{\gamma_0}\left(\rme^{\rmi \nu
t}\sigma_-\psi_1(t_{-})-\frac{1}{2}\,v(t)\psi_1(t_-)\right)\left(\sqrt{\epsilon}\,\rmd
\hat{W}_1(t)+\sqrt{1-\epsilon}\,\rmd \hat{W}_2(t)\right) ,
\end{multline*}
\begin{multline*}
\rmd\psi_2(t)=V_2(\psi_1(t_-),\psi_2(t_-))\rmd t +\left(\frac {\psi_1(t_{-}
)}{\norm{\psi_1(t_{-} )}}-\psi_2(t_{-} ) \right)\rmd N_2^1(t)\\ {}+\left(\frac
{\sigma_-\psi_1(t_{-} )}{\norm{\sigma_-\psi_1(t_{-} )}}-\psi_2(t_{-} ) \right)\rmd
N_1^1(t) -\psi_2(t_-)\rmd N_1^2(t)
\\
{}+ \sqrt{\gamma_0}\left(\rme^{\rmi \nu
t}\sigma_-\psi_2(t_{-})-\frac{1}{2}\,v(t)\psi_2(t_-)\right)\left(\sqrt{\epsilon}\,\rmd
\hat{W}_1(t)+\sqrt{1-\epsilon}\,\rmd \hat{W}_2(t)\right) ,
\end{multline*}
where $\hat W$ is the new Wiener process introduced in \eqref{hatW}, and
\begin{gather*}
v_1(t)=\sqrt{\gamma_0\epsilon}\,v(t), \qquad
v_2(t)=\sqrt{\gamma_0(1-\epsilon)}\,v(t),
\\ v(t)=2\sum_{k=1}^2\RE\left(\rme^{\rmi \nu t}\langle\psi_k(t_-)|\sigma_- \psi_k(t_-)\rangle\right),
\qquad I_2^1(t)=\varkappa \gamma_0\norm{\psi_1(t_-)}^2,
\\
I_1^1(t)=\gamma_1 \norm{\sigma_-\psi_1(t_-)}^2,
 \qquad I_1^2(t)= \gamma_2 \norm{\sigma_+\psi_2(t_-)}^2,
\end{gather*}
\begin{multline*}
V_j(\psi_1(t_-),\psi_2(t_-))=K^j\psi_j(t_-)+\frac{I_1^1(t) +I_1^2(t)+I_2^1(t)}{2}
\,\psi_j(t_-)\\
{}+\frac {\gamma_0} 2\, v(t)\sigma_- \psi_j(t_-)- \frac {\gamma_0} 4\, v(t)^2
\psi_j(t_-).
\end{multline*}

\subsection{The stochastic master equation and the
heterodyne spectrum}

In the situation we are considering the band transitions cannot be monitored. We
take under observation the system by collecting part of the emitted light in an
apparatus performing heterodyne detection. In this detection scheme the received
light is made to interfere with some monochromatic light of frequency $\nu$; to a
certain extent, this frequency can be varied. As we have said, it is $W_1$ which
represents the light reaching the detector; moreover, the (stochastic) output $J(t)$
of the detector is some smoothed version of $W_1$ \cite[Chapt.\ 7]{BarGreg09}, say
\begin{equation}\label{sm_out}
J(t)=\sqrt{k}\int_0^t \rme^{-k(t-s)/2}\, \rmd W_1(s), \qquad k>0.
\end{equation}

To take into account that only $W_1$ is observed we use the notation of Remark
\ref{rem:obs_out} and we take $d_1'=1$, $m_1'=d_1$, $m_2'=0$; recall that we have
$d_1=m_1=2$, $d_2=0$, $m_2=2$. Then, all the sums with jump processes disappear from
the stochastic master equations \eqref{xx} and \eqref{xxyy}. The linear stochastic
master equation \eqref{xx} becomes
\[
\rmd\sigma_j(t)=\mathcal{K}_j\big(\sigma_1(t),\sigma_2(t)\big)\rmd t
+\sqrt{\gamma_0\epsilon}\left( \rme^{\rmi \nu t}\sigma_-  \sigma_j(t)+ \rme^{-\rmi \nu
t} \sigma_j(t)\sigma_+\right)\rmd W_1(t),
\]
where the $\Kcal_j$ are the operators appearing in the Lindblad rate equations
\eqref{model}. The corresponding non linear stochastic master equation \eqref{xxyy}
for the \emph{a posteriori} states $\rho_j(t)= \sigma_j(t)\big/
\Tr_{\Hscr_S}\{\sigma_1(t)+ \sigma_2(t)\}$ turns out to be
\begin{multline*}
\rmd\rho_j(t)=\mathcal{K}_j\big(\rho_1(t),\rho_2(t)\big)\rmd t
\\ {}+\sqrt{\gamma_0\epsilon}\left( \rme^{\rmi \nu t}\sigma_-  \rho_j(t)+ \rme^{-\rmi \nu
t} \rho_j(t)\sigma_+ - m(t) \rho_j(t)\right)\rmd \hat W_1(t),
\end{multline*}
\[
m(t) = 2\RE \left( \rme^{\rmi \nu t}\Tr_{\Hscr_S} \left\{ \sigma_- \rho_j(t)\right\}\right).
\]

The power of the output current produced by the detector is proportional to $J(t)^2$
and the mean power at large times is proportional to
\begin{equation}\label{power1}
P(\nu)=\lim_{t\to +\infty}\Ebb_{\Pbb^t}[J(t)^2].
\end{equation}
The limit is to be taken in the sense of the distributions in $\nu$.

By using \eqref{sm_out} we get
\[
\Ebb_{\Pbb^t}[J(t)^2]=k\rme^{-kt}\Ebb_{\Pbb^t}\biggl[ \int_0^t \rme^{ks/2}\, \rmd W_1(s)
\int_0^t \rme^{kr/2}\, \rmd W_1(r)\biggr]
\]
So, to obtain an explicit expression for the power first of all we need to compute
the second moments of the Wiener type integrals $\int_0^t \rme^{ks/2}\, \rmd W_1(s)$
under the physical probability. The autocorrelation function of $W_1$, from which
such a moments follow, can be obtained by differentiation of the so called
\emph{characteristic operator} (the Fourier transform of the instruments)
\cite[Proposition 4.16]{BarGreg09}. The formula valid for the Markov case needs only
to be expressed by using the diagonal blocks; from \cite[Eq. (4.47)]{BarGreg09} we
get
\begin{multline}
\Ebb_{\Pbb^t}[J(t)^2]=k\int_0^t \rme^{-k(t-s)}\, \rmd s + 2k \gamma_0\epsilon
\int_0^t \rmd s \int_0^s \rmd r\, \rme^{-k(t-s)/2} \rme^{-k(t-r)/2}
\\ {}\times
\sum_{i,j=1}^2\Tr_{\Hscr_S}\left\{ \left(\rme^{\rmi \nu s} \sigma_-+ \rme^{-\rmi \nu
s} \sigma_+\right) \Tcal_{ij}(s-r)\left[\rme^{\rmi \nu s} \sigma_-\eta_j(r)+
\rme^{-\rmi \nu r}\eta_j(r)\sigma_+\right]\right\}.
\end{multline}
By $\sum_{j=1}^2\Tcal_{ij}(t)[\tau_j]$, $i=1,2$, we denote the solution of the
Lindblad rate equation \eqref{model} with initial condition $(\tau_1,\tau_2)$ at
time 0. Then, the computations needed to obtain $P(\nu)$ are long, but similar to
the ones in \cite[Sect.\ 9.1]{BarGreg09}; we give only the final results:
\begin{equation}\label{power2}
P(\nu)=1+4\pi \epsilon \Sigma(\nu),
\end{equation}
\begin{multline}
\Sigma(\nu)=  2\gamma_0\,\frac{\left(1-p\right)z_2^+\Gamma_2 -pz_1^+ w\left[
\Gamma_2 \left( \gamma_2- \gamma_1
-2\gamma_0\varkappa\right)+4\left(\omega_2-\omega_1\right)\left(\nu-\omega_2\right)\right]
}{\pi\left[4\left(\nu-\omega_2\right)^2+\Gamma_2^{\;2}\right]}
\\ {}+ 2\gamma_0\,\frac{pz_1^+ \left\{\left[ 1+w
\left( \gamma_2- \gamma_1 -2\gamma_0\varkappa\right)\right]
\Gamma_1+4w\left(\omega_2-\omega_1\right)\left(\nu-\omega_1\right)\right\}
}{\pi\left[4\left(\nu-\omega_1\right)^2+\Gamma_1^{\;2}\right]}\,,
\end{multline}
where $p$, $z_j^+$ have already be defined and
\[
\Gamma_1:=\gamma_0+\gamma_1+2\gamma_0\varkappa+k, \qquad
\Gamma_2:=\gamma_0+\gamma_2+k,
\]
\[
w:=\frac{2\gamma_0\varkappa}{4\left(\omega_1-\omega_2\right)^2+ \left( \gamma_2-
\gamma_1 -2\gamma_0\varkappa\right)^2}.
\]
In Eq.\ \eqref{power2} the term 1 is interpreted as the shot noise due to the local
oscillator and $\Sigma(\nu)$ as the heterodyne spectrum. Note that the widths
$\Gamma_j$ contain some dynamical parameters and the instrumental width $k$.

By its definition, we have $P(\nu)\geq 0$, while the positivity of the spectrum
$\Sigma(\nu)$ is not obvious. However, one can check that it is possible to rewrite
$\Sigma(\nu)$ in a form from which its positivity is apparent:
\begin{multline}
\Sigma(\nu)= D\gamma_0 \varkappa\Biggl\{
\frac{\gamma_0(1+\varkappa)+\gamma_1+k}{4\left(\nu-\omega_1\right)^2+\Gamma_1^{\;2}}
+ \frac{ \varkappa
\left(\gamma_2+k\right)}{4\left(\nu-\omega_2\right)^2+\Gamma_2^{\;2}}
\\ {}+ \frac{\gamma_0\varkappa\left(\Gamma_1+\Gamma_2\right)^2}{\left[4\left(\nu-\omega_1\right)^2+
\Gamma_1^{\;2}\right]
\left[4\left(\nu-\omega_2\right)^2+\Gamma_2^{\;2}\right]}\Biggr\},
\end{multline}
\[
D=
\frac{2/\pi}
{1+\varkappa \gamma_1/\gamma_2+ \varkappa(1+\varkappa)(1+\gamma_0/\gamma_2)}.
\]
Note that the heterodyne spectrum, for spontaneous emission in our model, contains
information on the dynamics: all the dynamical parameters, due to the structured
reservoir, determine the form of the spectrum. In particular, we have a double
peaked structure only if $\omega_1$ and $\omega_2$ are sufficiently different and
this difference can be only due to the band structure of the bath.

\section*{Acknowledgments} We thank Prof. F. Petruccione who introduced us to the
problem of unravelling Lindblad rate master equations.

CP acknowledges the financial support of the ANR ``Hamiltonian and Markovian
Approach of Statistical Quantum Physics'' (A.N.R. BLANC no ANR-09-BLAN-0098-01).

\end{document}